\documentclass[11pt]{article}
\usepackage{lmodern}
\usepackage{dsfont}
\usepackage{environ}
\input{definitions.sty}

\newcommand{\pmin}{{p_\text{min}}}
\newcommand{\pmax}{{p_\text{max}}}

\newcommand{\pt}{{\pi}} %
\newcommand{\ph}{{\hat{p}}}
\newcommand{\pz}{{p_0}}

\newcommand{\R}{{\mathbb{R}}}
\newcommand{\Rp}{{\mathbb{R}^+}}

\newcommand{\E}{\mathop{\mathbb{E}\/}}  %

\newcommand{\eps}{\epsilon}

\newcommand{\calD}{\mathcal{D}}

\newcommand{\calN}{\mathcal{N}}

\newcommand{\bone}{\boldsymbol{1}}

\def\<{\langle}
\def\>{\rangle}

\newcommand{\emailhref}[1]{\href{mailto:#1}{\tt #1}}

\begin{document}

\title{A Myersonian Framework for Optimal Liquidity Provision in Automated Market Makers}

\author{
 	    \textbf{Jason Milionis} \\
        \small Department of Computer Science \\
 		\small Columbia University \\
 		\small \emailhref{jm@cs.columbia.edu}
 		\and
 		\textbf{Ciamac C. Moallemi} \\
        \small Graduate School of Business \\
 		\small Columbia University \\
 		\small \emailhref{ciamac@gsb.columbia.edu}\\
 		\and
 		\textbf{Tim Roughgarden} \\
        \small Department of Computer Science \\
 		\small Columbia University \\ \small a16z Crypto \\
 		\small \emailhref{tim.roughgarden@gmail.com}
}
\date{Initial version: October 19, 2022 \\
      Current version: November 27, 2023
}
\maketitle

\thispagestyle{empty}

\begin{abstract}
In decentralized finance (``DeFi''), automated market makers (AMMs) enable traders to programmatically exchange one asset for another. Such trades are enabled by the assets deposited by liquidity providers (LPs). The goal of this paper is to characterize and interpret the optimal (i.e., profit-maximizing) strategy of a monopolist liquidity provider, as a function of that LP's beliefs about asset prices and trader behavior. We introduce a general framework for reasoning about AMMs based on a Bayesian-like belief inference framework, where LPs maintain an asset price estimate, which is updated by incorporating traders' price estimates. In this model, the market maker (i.e., LP) chooses a demand curve that specifies the quantity of a risky asset to be held at each dollar price. Traders arrive sequentially and submit a price bid that can be interpreted as their estimate of the risky asset price; the AMM responds to this submitted bid with an allocation of the risky asset to the trader, a payment that the trader must pay, and a revised internal estimate for the true asset price. We define an incentive-compatible (IC) AMM as one in which a trader's optimal strategy is to submit its true estimate of the asset price, and characterize the IC AMMs as those with downward-sloping demand curves and payments defined by a formula familiar from Myerson's optimal auction theory. We generalize Myerson's virtual values, and characterize the profit-maximizing IC AMM. The optimal demand curve generally has a jump that can be interpreted as a ``bid-ask spread,'' which we show is caused by a combination of adverse selection risk (dominant when the degree of information asymmetry is large) and monopoly pricing (dominant when asymmetry is small). This work opens up new research directions into the study of automated exchange mechanisms from the lens of optimal auction theory and iterative belief inference, using tools of theoretical computer science in a novel way.

\end{abstract}

\section{Introduction}

\newcommand{\nume}{num\'eraire\xspace}

\subsection{Exchanges and Market Makers}

The purpose of an exchange is to enable the trade of two or more
assets.  At a stock exchange, for example, shares of a stock might be
exchanged for US dollars via a mechanism known as a ``limit
order book'' (LOB), in which buyers and sellers post ``limit orders''
(offers to buy or sell a specified quantity at a specified price)
that are then matched greedily.

Every trade has two sides: for one user to buy a risky asset at a
given price, there must be a corresponding seller willing to sell at
that same price.  Such trades can occur organically due to a
coincidence of wants from two traders, but in practice many trades are
enabled by professional {\em market makers} who continually match buy
and sell orders whenever the price is right. This paper is about the
profit-maximization problem faced by such a market maker.
In particular, the market maker is continually updating their beliefs around the asset price, according to a Bayesian-like belief inference framework, from a prior belief in conjunction with observations by traders who report price estimates.

Digital assets secured by blockchains, such as cryptocurrencies, are
also traded via exchanges.
In centralized exchanges such as Coinbase or Binance, users hand control of their assets over to the exchange and, therefore, accept the credit risk of not getting them back.
Such exchanges are typically based on the CLOB design.
Decentralized exchanges (DEXs), on the other hand, operate purely programmatically (i.e., ``on-chain''),
are typically non-custodial (meaning that traders at all times have direct control of their assets
in the sense that assets are not entrusted to a third party), and often depart from the CLOB
model.  There are several reasons for this departure: on-chain computation and storage can be
expensive, and CLOBs need a lot of both; the markets for many digital assets (especially long-tail
ones) are thin and have no dedicated market makers, in which case CLOBs may fail to provide an
acceptable level of liquidity; and a general openness in the blockchain world to experimental
designs.  The DEX design space is large, and many different designs have been deployed (on
Ethereum and other general-purpose smart contract platforms) over the past years.

For example, Uniswap (especially in its earliest iterations, v1 and v2) is
a canonical example of an {\em automated market maker (AMM)}.  In an
AMM (unlike in a CLOB), the roles of traders and market makers are
clearly separated.  A market maker, which in this context is called a
{\em liquidity provider (LP)}, deposits into the AMM some amounts of
the assets being traded (and is, by default, passive thereafter).  At
all times, the AMM defines a marginal price as a function of its current
reserves of those assets.  (In the case of Uniswap v1 and v2, the spot
price of an AMM is simply the ratio of the quantities of its two
assets.)  The AMM is willing to take either side of a trade (buy or
sell) at its current marginal price, and in this sense is always providing
liquidity to traders.  The latest iteration of Uniswap (v3) blurs the
line between AMMs and CLOBs by allowing a much wider range of LP
demand curves than in v2 (to some extent approximating what can be
expressed by limit orders in a CLOB), thereby encouraging an LP to
more actively manage its demand curve as market conditions evolve.

The goal of this paper is to characterize the solution to the
following fundamental question: What is the optimal liquidity demand curve for a
profit-maximizing market maker?  This question is relevant for a
market maker in a CLOB, an LP of an expressive AMM such as Uniswap v3,
or even the designer of a new AMM (in which case the market making
liquidity demand curve is defined programmatically rather than submitted by a
third-party LP).  This paper offers a new model for reasoning about
this question and a quite general answer, which resembles in many
respects, and in fact generalizes, Myerson's theory of optimal auctions~\parencite{myerson_optimal_auctions}. To get
a feel for our model and results, we next discuss a highly restricted
but nonetheless illuminating special case, and give a very brief overview of our general theory.

\subsection{The Pure ``Noise Trading'' Case and an Overview of the General Case}
\label{subsec:intro_specific}

The theory of market microstructure (see e.g.,~\textcite{ohara})
differentiates between ``informed traders,'' who have better
information about the fundamental value of an asset than the market maker
does, and ``noise traders,'' who trade for idiosyncratic reasons
that do not reflect any private information about the fundamental value of
the asset.
Generally speaking, market makers lose money to informed traders due
to adverse selection: by virtue of such a trader being willing to
trade, it expects a profit, which, due to the zero-sum nature of
common-value assets, comes at the expense of the market maker.
Market makers can, however, generally profit from noise traders.
In this section, we first consider (as an instructive warm-up) the special case in which {\em all} traders are noise traders (and hence adverse selection plays no role).  Our general model, the results of which are briefly presented in the end of this section, and which is detailed in Section~\ref{sec:bayes} accommodates
arbitrary mixtures of informed and noise traders, and our general
results must therefore reflect adverse selection effects in addition
to the aspects identified in this section.

Fundamentally, a market maker must decide, as a function of the
information available to it, how much of a risky asset they are
willing to buy or sell at different prices.  For this example, we will
assume that the market maker knows the true value~$p_0$ of the risky
asset that it is trading.

An obvious strategy for such a market maker is to accept exactly the
orders on which it makes money: all (and only) the buy orders
with offered price more than~$p_0$, and all (and only) the sell orders
with offered price less than~$p_0$.
However, as we will see, this is not in general the profit-maximizing strategy
for a market maker.

It will be convenient to encode a market maker's strategy via the choice
of a {\em demand curve}~$g$, which maps prices to quantities (i.e., an
amount of risky asset that will be held by the market maker).  The
interpretation is that the market maker would be interested in selling
$g(p_0)-g(p)$ units of the risky asset at the price~$p$.  (Note
that this may be a negative number of units, in which case the market
maker will be buying the risky asset from the trader.)  A choice
of~$p_0$ and~$g$ induces the {\em allocation rule} $x(p) =
g(p_0)-g(p)$, the amount of the risky asset allocated to a trader who
reports a price of~$p$.

Next, we assume that a trader reports a price $\ph \sim \calD$
that is drawn according to a distribution~$\calD$ known to the market
maker.  We assume that the trader is willing to buy or sell at most a
bounded amount of the risky asset (normalized to be~1 only for notational convenience, since there is no difference whatsoever for any arbitrary finite bound on the results), and that the
demand curve~$g$ and price~$p_0$ accordingly satisfy
$g(p_0)-g(\ph) \in [-1,1]$ for all~$\ph$.  When the noise trader reports
their price~$\ph$, they receive an allocation $x(\ph)=g(p_0)-g(\ph)$
of the risky asset and must remit a payment~$y(\ph$) to the AMM for
this provision. As we will see in Section~~\ref{sec:ic_amms}, the
natural notion of incentive-compatibility in this setting --- that a
utility-maximizing trader who values the risky asset at a price~$\ph$
will in fact report~$\ph$ to the AMM --- uniquely pins down the payments
(the $y(\ph)$'s) as a function of the allocation rule (the
$x(\ph)$'s).\footnote{As the reader might guess, it is also essential
  that the allocation rule~$x$ is monotone in~$\ph$, or equivalently
  that the demand curve~$g$ is downward-sloping, with the market maker
  selling more and more of the risky asset as the going price gets
  higher and higher.  See Section~\ref{sec:ic_amms} for details.}

With informed traders, a rational market maker will update its
estimate~$p_0$ of the risky asset's true price following a trade
(generally increasing it after a trader buys the risky asset from the
market maker and decreasing it after a sell).  Our general model in
Section~\ref{sec:bayes} accommodates this through a generic update
rule~$\pi(p_0,\ph)$ that specifies the market maker's new estimate as
a function of its old estimate~$p_0$ and the new information conveyed
by a trader reporting the price~$\ph$.
For example, $\pi$ might
correspond to a Bayesian update with respect to some assumed
information structure, and the update rule may come from and be revised according to a no-regret learning algorithm used by the market maker on successive rounds.
In the simple special case of only noise traders being present, because the behavior
of noise traders is (by assumption) independent of the true value of
the risky asset, the market maker has no cause for updating its
internal price (i.e., $\pi(p_0,\ph)=p_0$ for all $p_0$ and $\ph$).

What is the optimal strategy (i.e., choice of demand curve~$g$) for
the market maker?  That is, which choice of~$g$ maximizes
\begin{equation}\label{eq:noise}
\E[\text{Profit}] = \E[y(\ph) - \pz \cdot x(\ph)],
\end{equation}
where the expectation is with respect to $\ph \sim \calD$,
$x(\ph) = g(\pz)-g(\ph)$ is the allocation rule, and~$y$ is the
uniquely defined payment rule mentioned above?\footnote{Note that, obviously, the choice of the optimal mechanism does not need the knowledge of the reported price by the trader; only its probabilistic characteristics are known according to the distribution $\calD$.} (In the general case,
the ``$p_0"$ term in~\eqref{eq:noise} is replaced by $\pi(\pz,\ph)$,
with the market maker evaluating its portfolio according to its new
belief about the value of the risky asset, conditioned on observing the trader's report of $\ph$.)

The results in this paper imply that the answer to this question is
always defined by a posted-price-like mechanism with two prices (depending on $\calD$), $p_l \le p_0$ and
$p_h \ge p_0$.  These two prices split the price range into three
intervals, and the optimal strategy is to always buy the maximum
amount in the lowest price interval, always sell in the highest price
interval, and refuse to trade in the middle price interval.

We will be interested in the length of this middle price interval,
which we call the {\em no-trade gap}.  Bigger no-trade gaps are
generally worse from a welfare perspective, with welfare-increasing
trades sacrificed in the name of higher profit to the market maker.
The no-trade gap is reminiscent of the ``bid-ask spread'' in CLOBs,
which occupies a central position in the study of traditional
market-making in CLOBs.\footnote{In a CLOB, the ``bid'' and ``ask''
  refer to the highest price of an outstanding buy order and the
  lowest price of an outstanding sell order, respectively.  Because
  limit orders in a CLOB are always matched greedily, the bid-ask
  spread is always non-negative.  Small spreads are generally viewed as
  a good thing, indicating a ``highly liquid'' market (with lower trading costs, at least for small quantities).}
Traditionally, theoretical justifications for non-trivial bid-ask
spreads in CLOBs have relied on frictions (like transaction fees or
inventory costs) or adverse selection (with a market maker needing to
exploit noise traders to cover the losses to informed traders).
In the current example, with neither frictions nor adverse selection,
the no-trade gap arises for a different reason, namely the monopolist
position of the market-maker.  In our general model, with the
possibility of informed traders, monopoly pricing and adverse
selection both contribute to the no-trade gap.  For example, no matter
what the distribution $\calD$ is, if all traders have perfect information
(i.e., extreme adverse selection), then the no-trade gap encompasses
the full price range.

The prices~$p_l$ and~$p_h$ that split the price range into three
intervals, and hence the no-trade gap, are distribution-dependent.
For example, intuitively, one might expect a smaller gap for
distributions~$\calD$ that are tightly concentrated around~$p_0$.
In general, our analysis shows that these two prices can be
characterized as roots of two functions that resemble the virtual
valuation functions used in optimal auction theory~\parencite{myerson_optimal_auctions}.

We next give two example cases for pure noise-trading: if $\calD$ is the uniform
distribution on $[\pmin,\pmax]$, then $p_l = (\pmin+\pz)/2$ and
$p_h = (\pz+\pmax)/2$; the resulting optimal allocation rule is depicted in \Cref{subfig:uniform}.
If $\calD$ is the exponential distribution with parameter $\lambda$ (on $[0,\infty)$), then~$p_h = p_0 + \tfrac{1}{\lambda}$.
The lower price~$p_l$ does not have a closed-form expression in this case, but can be easily approximated numerically (given $\pz$ and $\lambda$); \Cref{fig:table_exponential} enumerates some examples of how these lower prices vary, and an illustration of the optimal allocation rule is contained in \Cref{subfig:exp}.

\begin{figure}[h]
\centering
\captionsetup[subfigure]{labelformat=empty}
\hspace*{\fill}
\begin{subfloat}[\label{subfig:uniform}]
\centering
\begin{tikzpicture}
\begin{axis}[
    clip=true,scale=1.5,
    axis y line=left,
    axis x line=middle,
    xtick={1.1, 1.55, 2.0},
    xticklabels={$p_0$, $\frac{p_\text{max}+p_0}{2}$, $p_\text{max}$},
    extra x ticks={0.2, 0.65},
    extra x tick labels={$p_\text{min}$, $\frac{p_\text{min}+p_0}{2}$},
    extra x tick style={
        xticklabel style={
            anchor=south,
        }
    },
    ytick={-1, 0, 1},
    yticklabels={-1, 0, 1},
    xlabel=$\hat{p}$,
    xmin=0, xmax=2.2,
    ymin=-1.5, ymax=1.5,
    footnotesize,
    ];
    \addplot[mark=none,blue,thick,domain=0.2:0.65]{-1};
    \addplot[mark=none,blue,thick] coordinates {(0.65, -1) (0.65, 0)};
    \addplot[mark=none,blue,thick,domain=0.65:1.55]{0};
    \addplot[mark=none,blue,thick] coordinates {(1.55, 0) (1.55, 1)};
    \addplot[mark=none,blue,thick,domain=1.55:2.0]{1} node[midway,above] {$x_a(\hat{p})$};
    \addplot[mark=none,dashed] coordinates {(0.2, -1) (0.2, 0)};
    \addplot[mark=none,dashed] coordinates {(2, 0) (2, 1)};
\end{axis}
\end{tikzpicture}
\end{subfloat}
\hfill
\begin{subfloat}[\label{subfig:exp}]
\centering
\begin{tikzpicture}
\begin{axis}[
    clip=true,scale=1.5,
    axis y line=left,
    axis x line=middle,
    xtick={1, 1.5},
    xticklabels={$p_0$, $p_0 + \frac{1}{\lambda}$},
    extra x ticks={0.396},
    extra x tick labels={$p_l$},
    extra x tick style={
        xticklabel style={
            anchor=south,
        }
    },
    ytick={-1, 0, 1},
    yticklabels={-1, 0, 1},
    xlabel=$\hat{p}$,
    xmin=0, xmax=2.2,
    ymin=-1.5, ymax=1.5,
    footnotesize,
    ];
    \addplot[mark=none,blue,thick,domain=0:0.396]{-1};
    \addplot[mark=none,blue,thick] coordinates {(0.396, -1) (0.396, 0)};
    \addplot[mark=none,blue,thick,domain=0.396:1.5]{0};
    \addplot[mark=none,blue,thick] coordinates {(1.5, 0) (1.5, 1)};
    \addplot[mark=none,blue,thick,domain=1.5:2.2]{1} node[midway,above] {$x_b(\hat{p})$};
\end{axis}
\end{tikzpicture}
\end{subfloat}
\hspace*{\fill}
\caption{Optimal allocation rule $x(\hat{p})$ in ``pure noise trading'' showing the no-trade gap around $p_0$ when the distribution $\hat{p}\sim\mathcal{D}$ is: (a) the uniform distribution on $[p_\text{min}, p_\text{max}]$, (b) the exponential distribution on $[0, \infty)$ with parameter $\lambda=2$ for $p_0=1$ (the computation of $p_l$ can only be done via numerical methods; see \Cref{fig:table_exponential}).}
\label{fig:ex}
\end{figure}

\begin{figure}[h]
\centering
\begin{tabular}{|c|c|c|c|}
\hline
$p_0$ \textbackslash{} $\lambda$ & 0.5   & 1     & 2     \\ \hline
0.25                             & 0.123 & 0.121 & 0.118 \\ \hline
0.5                              & 0.242 & 0.235 & 0.221 \\ \hline
0.75                             & 0.358 & 0.342 & 0.314 \\ \hline
1                                & 0.470 & 0.443 & 0.396 \\ \hline
1.5                              & 0.684 & 0.627 & 0.537 \\ \hline
2                                & 0.886 & 0.792 & 0.653 \\ \hline
\end{tabular}
\caption{Table of lower prices $p_l$ in ``pure noise trading'' when the distribution $\hat{p}\sim\mathcal{D}$ is the exponential distribution on $[0, \infty)$ with parameter $\lambda$.}
\label{fig:table_exponential}
\end{figure}

The rest of this paper develops a similarly precise understanding of
optimal market-maker strategies in a more general setting.
We will now briefly present here the general results in an informal way.
For the details, the reader is referred to \Cref{sec:bayes}.

\begin{inftheorem}
There exist virtual valuation functions $\phi_u(s)$ and $\phi_l(s)$, such that the expected profit, shown in \Cref{eq:noise} in the general case where
the ``$p_0"$ term is replaced by $\pi(\pz,\ph)$, is equal to the expected virtual welfare of an auction with virtual values derived by $\phi_u(s)$ if $s\ge\pz$ and $\phi_l(s)$ if $s < \pz$.
\end{inftheorem}

\begin{inftheorem}
There exist two prices $p_l \le \pz$ and $p_h \ge \pz$, such that the optimal allocation rule of a market maker that maximizes their expected profit is to buy the maximum amount if the trader reports a price $\ph \le p_l$ and sell the maximum amount if the trader reports a price $\ph \ge p_h$, and refuse to trade in the interval between these two prices, i.e., when $\ph \in (p_l, p_h)$.
\end{inftheorem}

The formal statements of the above results are \Cref{thm:exp_virtual_welfare} and \Cref{thm:opt_alloc_rule} respectively.
Finally, in \Cref{sec:lin_lambda}, we examine how our general theory above applies in the interesting special case of a linear update rule.
Through this special case, we naturally show that the structure of the no-trade gap captures precisely the complex effect that is the interpolation between the effects of adverse selection (observed when all traders have perfect information) and monopoly pricing (observed when all traders are noise traders).
For further details, we refer the interested reader to \Cref{sec:lin_lambda}.

\subsection{Related Work}
\label{subsec:litrev}

\paragraph{Auctions.}
This paper critically relates to revenue maximizing auction design, initiated under the pioneering
work of \textcite{myerson_optimal_auctions}.  Since then, various versions of virtual values have
been used in a variety of settings among others for specifying optimal auctions and their
approximations, including as to how they relate to no-regret learning in auctions
(e.g.,~\textcite{bulowroberts,chawla2007algorithmic,hartline2012approximation,haghpanah2015reverse,CHAWLA201980,chawla2010multi,alaei2013simple,roughgarden2019approximately,roughgarden2016optimal,hartline2009simple});
in a similar spirit, we define and use functions that can be interpreted as virtual valuations.  For
a survey of results around revenue-optimal mechanisms, virtual values, optimal auctions, and their
approximations, see the book by \textcite{hartlinebook}.

\paragraph{AMMs.}
Constant function market makers (CFMMs) have arisen from the idea of holding some
function constant across ``trading states'' (from states of the world, previously, on the
utility-based framework of AMMs for prediction markets by \textcite{chenpennock2007}), and are characterized by such an
invariant called a \textit{bonding function}.
This means that the AMM is willing to take the other side of any trade that corresponds to remaining on a constant level curve of the particular bonding function chosen.
Tools from convex analysis ---familiar to both the algorithmic game theory and machine learning theory communities--- are crucially used to analyze CFMMs and prove ``optimal behavior'' properties for participating traders \parencite{jasonLVR_CCSDeFi,jason_arb_profits,angeris2020improved,angeris2021replicatingmarketmakers,angeris2021replicatingmonotonicpayoffs}.

The description of the framework that we use for AMMs --- inspired by a reparameterization of a
CFMM curve (established by \textcite{angeris2020improved}) in terms of portfolio holdings of the pool
with respect to the price --- resembles that in \textcite{exchange_complexity}, who used a
similar AMM framework for a different purpose, namely
unifying constant function market makers
(CFMMs) and limit order books (LOBs) and 
complexity-approximation trade-offs between them. 
The latter work considers neither incentive-compatibility constraints nor
the problem of optimal liquidity provision.

\paragraph{Asset price beliefs in AMMs.}
\textcite{mazieres} also examine beliefs of LPs (specifically on CFMMs)
around future asset prices; however, in their model, these beliefs are
static (i.e., similar to a single prior that remains unchanged
throughout the trades' occurrence) or modeled by a fixed price dynamic
discounted to the present (i.e., still not dynamically revised in
response to incoming trades).  The focus of \textcite{mazieres} is the design of
CFMMs that maximize the fraction of trades
that the CFMM can satisfy with small slippage under the stationary
distribution of the Markov chain induced by their model. 
Our work focuses instead on LP profit-maximization with
informed traders and noise traders, showing a formal relation between the optimal auction design literature, including computational considerations in approximations of those, and the market microstructure literature.

There is also work specific to Uniswap v3 from the LP perspective around how beliefs about future prices should guide the choice of an LP's demand curve \parencite{fan2022differential,yin2021liquidity,neuder2021strategic}.

\paragraph{Market microstructure.}
There is a large literature on the market microstructure of limit order books; see the textbook by \textcite{ohara} and references therein.
Closest to the present work is the paper of \textcite{glosten1989insider}. Focusing on the role of asymmetric information, \textcite{glosten1989insider} considers the problem of a monopolist optimizing their liquidity demand curve, but in the presence of risk averse traders (in contrast to our setting, where the traders can be viewed as risk neutral).
The risk averse setting is more technically challenging, and hence \textcite{glosten1989insider} is required to make very strong and highly specific technical assumptions (normally distributed values and observation noise) that are not necessary, obscure the general intuition in our setting, and do not add value.
Instead, the intuition provided by our model is much stronger and more general.
Moreover, our risk neutrality assumption unlocks a fruitful connection between optimal market-making and the Myersonian theory of optimal auction design.
\textcite{glosten1994electronic} considers the related problem of the liquidity demand curve in a competitive equilibrium.
The latter paper is restricted to a limit order book, and does not consider more general forms of liquidity provision in exchange mechanisms, like the ones we do.
\textcite{glostenmilgrom1985} operate in a competitive regime with zero expected profits, which gives a completely different reason for bid-ask spreads to arise than the one given in our work.
\textcite{kyle1985} develops a model where a \emph{single} informed trader (``insider'') competes with noise traders and market makers, showing the incorporation of insider information into the equilibrium price. In this work, the focus is on the informed trader who is a monopolist (there is only one), whereas in our work informed traders are unrestricted and the market maker is a monopolist in their liquidity provision.
\textcite{kyle1989} continues to examine the informational efficiency of equilibrium prices in a slightly more general setting, where there is competition among informed traders, showing that prices are less informative than in the competitive equilibrium.
Both of these papers use strong assumptions, predominantly the normal distribution over both the values and the order sizes.
While their models are related to ours in their considerations of types of participants present and observations in the optimal case, the present results were not known in either of these two works; neither of these models is as general as ours, neither relates whatsoever to mechanism design and incentive compatibility constraints, and both are based on very strong technical assumptions that our work does not need.

\paragraph{Sequential market-making.}
Finally, the framework we describe is similar to sequential market-making procedures in its use of iterative belief updating.
Such procedures have been studied in the past including in the context of online learning frameworks \parencite{nips2013_abernethy_ol_mm,bayesian_mm_EC2012,nips2008sequential_marketmaking,das2005}.
Closer to our Bayesian setting, \textcite{nips2008sequential_marketmaking} consider the liquidity present across trade sizes in a multi-period monopolistic market making model (and subsequently shows the superiority of this compared to a competitive multi-dealer market).
The above works, however, do not show a connection between optimal market-making and the Myersonian theory of optimal auction design, which is the key focus in our results.

\section{Automated market makers}
\label{sec:ic_amms}

We begin by presenting a general framework of exchange for automated market makers (AMMs) that allows us to consider market making as conducting an auction, and is based on market participants (traders) trading with the AMM on the basis of future price beliefs.

More specifically, suppose there are two assets, a risky asset and a \nume asset. An AMM according
to this framework is defined by a \textbf{demand curve} $g:(0,\infty)\to\Rp$ such that the
function $g(\cdot)$ represents the \emph{current} demand of the market maker for the risky asset
as a function of price, i.e., the amount of risky asset that the market maker wishes to hold at each possible price,
along with a payment rule $y:(0,\infty)\to\R$ such that the function $y(\cdot)$ specifies the quantity of \nume that a trader needs to pay depending on the \emph{current} AMM price $p_0$ (which instantiates a specific state of this AMM's demand curve, i.e., a current demand of the AMM for $g(\pz)$ of risky asset) and what they declare to the AMM that their price estimate is (see the following paragraph on trading for the precise definition).
The market maker is assumed to have no inventory costs or budget constraints.

The demand curve $g$ along with an appropriately defined payment rule can arise through \emph{bonding curves} of traditional constant function market makers (CFMMs, i.e., functions $f$ such that the holdings of the joint pool $(x,y)$ satisfy $f(x,y)=c$ for some $c$) but this is not necessary: the exchange mechanisms defined by this framework strictly generalize CFMMs (see also \textcite{exchange_complexity}).
Additionally, a non-continuous demand curve $g$ can also be used to express limit orders, with the discontinuities representing limit buy or limit sell orders at the corresponding price, of quantity of risky asset equal to the jump at that price \parencite{exchange_complexity}.

\paragraph{Trading} A trader who wants to trade with the AMM has some (hidden) belief about the asset price. The trader will come to the mechanism and specify a price bid $\ph$. Based on their bid, the trader will get a quantity $g(\pz) - g(\ph)$ of risky asset.
Finally, the AMM revises its current price (which was previously $\pz$) according to some update rule (we will see more about the update rule from \Cref{sec:bayes} onwards; it will play no role in this section).
We say that the trader was \textit{allocated} a quantity $x(\ph) \triangleq g(\pz) - g(\ph)$ of asset.
This interpretation makes it natural to consider the transaction with the trader \textbf{as an auction} with one bidder; the trader has a price estimate (valuation) for the risky asset, will be allocated a portion of the risky asset (provided by the market maker) based on their submitted price bid, and will need to pay the market maker an amount $y(\ph)$ of \nume for it. Notice that it may be the case that $x(\ph) < 0$, which means that the trader will sell risky asset to the AMM, and the market maker will now need to come up with the corresponding payment to the trader for it.

\paragraph{Example: Uniswap v2} Let's consider the simple case of a constant product market maker (CPMM), such as Uniswap v2, to show how it emerges from the framework above. Consider the demand curve of the AMM to be of the form
$
g(p)=\frac{c}{\sqrt{p}}
\,,
$
for some $c > 0$, and the payment rule to be of the form
$
y(\ph)= c \left( \sqrt{\ph} - \sqrt{\pz} \right)
\,,
$
for the same $c > 0$. A trader who will trade with this exchange at a current price $p_0$ with a price bid $\ph$ will obtain a quantity $x(\ph) = g(p_0) - g(\ph) = c\left( \frac{1}{\sqrt{p_0}} - \frac{1}{\sqrt{\ph}} \right)$ of risky asset.
Since the curve $g(p)$ above is just a reparameterization (in terms of prices) of the risky asset
holdings of the CPMM curve\footnote{Note that here $x$ and $y$ --- without the parentheses --- are
  not the functions above, but rather the current total reserves of the AMM pool.} $xy=k$ where $k=c^2$
\parencite{angeris2021replicatingmonotonicpayoffs},
by comparing the expression of the quantity of risky asset that the trader gets, we conclude that they get the same quantity as in the CPMM.

\subsection{Incentive-compatible AMMs}

We will now consider that traders interacting with AMMs have utility functions that emerge naturally from their \emph{true} belief about the asset price versus what they bid to obtain a quantity of risky asset from the AMM. In particular, a trader who has a \emph{true} belief of $p$ for the asset's price, and submits a bid $\ph$ to the AMM seeks to optimize (maximize) their utility defined by
\begin{align}
\label{def:trader_util}
u(p, \ph) \triangleq p\cdot x(\ph) - y(\ph)
\,,
\end{align}
because they consider that the asset's true value to them is $p$ per unit, they obtain a quantity $x(\ph)$ of it, and pay $y(\ph)$ for this trade.

\begin{definition}
\label{def:ic_amm}
[Incentive-compatible AMM]
An AMM defined according to a demand curve $g(p)$ and a payment rule $y(\ph)$ is called \emph{incentive compatible} (IC), if any trader whose utility follows \Cref{def:trader_util} has as an optimal strategy for interacting with the AMM to submit their \emph{true} belief about the asset price as their bid to the AMM.
\end{definition}

Is an AMM defined by an arbitrary demand curve $g(p)$ and an arbitrary payment rule $y(\ph)$ incentive compatible? The answer is no, and the characterization of IC AMMs follows standard results of dominant-strategy incentive-compatible (DSIC) single-parameter auctions.

\begin{proposition}
[Characterization of IC AMMs]
\label{prop:ic_amms_non_incr_g}
An AMM defined according to a demand curve $g(p)$ can be paired with a payment rule $y(\ph)$ to obtain an IC AMM, if and only if $g(p)$ is a non-increasing function.
\end{proposition}
\begin{proof}
First, we observe the analogy drawn to \textit{allocations} of asset for trades: $x(\ph) = g(\pz) - g(\ph)$.
Since the utility function for the auction with a single bidder (the trader) is of the form of \Cref{def:trader_util}, the proposition follows from the characterization of dominant-strategy incentive-compatible (DSIC) single-parameter auctions' allocation rules $x(\ph)$ as only those functions which are monotone, and specifically non-decreasing functions, because $x(\ph)$ is non-decreasing precisely if and only if $g(\ph)$ is non-increasing.

More specifically, we follow the standard analysis to showcase the exact correspondence of the settings but do not show the entire analysis, as this follows Myerson's. In order for the optimal strategy of any trader with true belief $p$ to be to submit $p$ to the AMM, we need to have that for any $p,\ph$:
\[
u(p,p) \ge u(p,\ph)
\Leftrightarrow
p\cdot x(p) - y(p) \ge p\cdot x(\ph) - y(\ph)
\Leftrightarrow
p\cdot\left[ x(p)-x(\ph) \right] \ge y(p) - y(\ph)
\,.
\]
Additionally, if the converse holds, and the true belief of the trader was $\ph$ while it is considering submitting $p$ to the AMM, then we also need to have that for any $p,\ph$:
\[
u(\ph,\ph) \ge u(\ph, p)
\Leftrightarrow
\ph\cdot x(\ph) - y(\ph) \ge \ph\cdot x(p) - y(p)
\Leftrightarrow
\ph\cdot\left[ x(p) - x(\ph) \right] \le y(p) - y(\ph)
\,.
\]
Combining these two inequalities, we get that for any $p,\ph$:
\[
\ph\cdot\left[ x(p) - x(\ph) \right] \le y(p) - y(\ph) \le p\cdot\left[ x(p)-x(\ph) \right]
\,,
\]
hence it is immediate (from the derived inequality $(p-\ph) \cdot\left[ x(p) - x(\ph) \right] \ge 0$) that $x(\ph)$ needs to be non-decreasing.
\end{proof}

The next question to be answered is whether arbitrary payment rules $y(\ph)$ are allowed to be requested by an IC AMM.
Not surprisingly, as is the case for payments defined by an allocation rule of a DSIC mechanism, the payments need to be precisely defined by the demand curve for an IC AMM. The below formula is familiar from Myerson's single-parameter DSIC auction design.

\begin{corollary}
[Payments of traders on IC AMMs]
Suppose that any trader who is allocated a quantity $0$ of risky asset pays $0$ amount of \nume to an IC AMM according to \Cref{def:ic_amm} with a demand curve $g(p)$. The IC AMM's payment rule $y(\ph)$ for a trader submitting a price bid $\ph$ needs to be exactly:
\begin{equation}
\label{eq:ic_amm_payment}
y(\ph)
=
\int_\pz^\ph s ~dx(s)
=
-\int_\pz^\ph p ~dg(p)
\,.
\end{equation}
\end{corollary}
The above integrals are Riemann--Stieltjes integrals. In cases where $g(p)$ is differentiable, the differential takes the form $dg(p) = g'(p)~ dp$.
\begin{proof}
From Myerson's single-parameter auction design, we know that there is a unique payment rule (up to an additive constant) for DSIC single-parameter auctions; that is,
$
y(\ph)
=
\int_\pz^\ph s ~dx(s)
\,,
$
where the arbitrary additive constant is no longer present, as we have integrated the required normalization that for an allocation of $x(\pz) = g(\pz)-g(\pz)=0$ asset, the payment needs to be $y(\pz)=0$.
This directly translates to the formula given in the corollary's statement in terms of the demand curve $g(p)$ of the AMM.
\end{proof}

Note that the integral that defines the payments is non-negative if $\ph \ge \pz$, and non-positive if $\ph \le \pz$. This is compatible with the direction of the trade, as defined by the allocation rule and \Cref{prop:ic_amms_non_incr_g}, which is that $x(\ph) \ge 0$ for $\ph \ge \pz$, and $x(\ph) \le 0$ for $\ph \le \pz$.

\section{Optimal liquidity provision via revenue maximization}
\label{sec:bayes}

\subsection{Bayesian model for optimal allocation rule}
\label{subsec:model_optimal_alloc}

First, we describe our model for how the market maker should optimize their allocation rule for the asset.
From now on, the market maker will be assumed to be risk-neutral, i.e., be indifferent to ``risk,'' and only care to be profit-maximizing.

The Bayesian belief framework operates as follows: the market maker has a prior belief distribution on the asset price, and the (single) price it decides to use is $\pz$; for example, and most commonly, $\pz$ will be the mean of this prior distribution. A trader then comes to the AMM with an estimate $\ph$ (that is the true belief of the trader, if the AMM is IC according to \Cref{def:ic_amm}).
Traders are on the market to buy or sell up to 1 unit of the asset;\footnote{This is just for notational convenience. In fact, it generalizes easily for demand up to any arbitrary finite bound, and a subsequent normalization enables using exactly our results here with unit demand.} that is, the allocation rule should satisfy $-1\le x(\ph) \le 1$ for all $\ph$.
The trader's estimate $\ph$ is perceived by the market maker as one observation of an asset price estimate coming from a distribution $\calD$ that is included in the market maker's model.
For this distribution $\calD$, we assume that it has a continuous, positive density function $f(\ph)$ on a compact support $[\pmin,\pmax]$ and a cumulative density function (CDF) $F(\ph)=\int_\pmin^\ph f(s)~ds.$\footnote{Our results can be extended to non-continuous density functions supported at possibly non-compact intervals, e.g., $(0,\infty)$, but for simplicity of exposition, we do not focus on handling these edge cases here, since they do not add to the general discussion points of the paper.}

Applying Bayes' theorem allows the market maker to specify the posterior distribution of the asset price, conditional on the observation from the trader.
The market maker then sets this posterior distribution as their new prior for future observations (trades).
We will assume that the way this posterior is used by the AMM is through valuing the quantity of asset given or obtained at a price of $\pt(\pz,\ph)$ which can be, for instance, the mean of the posterior distribution according to the Bayesian framework.
(A very interesting special case as to how a linear dependence arises on the mean of the posterior with normally distributed prior and observational error is shown in \Cref{sec:lin_lambda}.)

In general, we will call $\pt(\pz,\ph)$ the ``update rule'' and --- for generality --- we will not impose further constraints arising from the Bayesian framework on what it may be other than that:

\begin{assumption}
\label{ass:bayesian}
The following conditions have to hold for the update rule:
\begin{itemize}
\item $\pt(\pz, \ph)$ lies in the interval between $\pz$ and $\ph$.%
\item $\pt(\pz, \ph)$ is non-decreasing in $\ph$.
\item $\pt(\pz, \pz) = \pz$, i.e., if the trader's estimate is the same as the market maker's prior, then the updated valuation of the market maker for the asset should be the same as their prior.
\item $\E_{\ph\sim\calD}[\pt(\pz, \ph)] = \pz$, i.e., the prior is consistent with the update rule.
\end{itemize}
\end{assumption}

Therefore, according to the above, the market maker's next best estimate for the asset's price is $\pt(\pz,\ph)$ conditioned on receiving the trader's estimate $\ph$.
It is then natural to consider that, when selling (respectively, buying) some quantity $x(\ph)$ of the asset from the trader, the AMM will not only consider the amount of money that they receive (respectively, give) but also their perceived best estimate of the value of the asset that they lost (respectively, gained).
This, then, arises as an interesting contrast to traditional auction theory, where the revenue of the auction is only the payment that is made.
Here, according to the above reasoning, it makes sense for the AMM to value their profit according to
\begin{align}
\label{eq:revenue}
\text{Profit} = y(\ph) - \pt(\pz, \ph) \cdot x(\ph) \, .
\end{align}
Therefore, an AMM choosing an optimal strategy to provide liquidity will consider the strategy that maximizes their profit, among all feasible, incentive-compatible (for the traders) strategies.

\subsection{Equivalence of expected profit and virtual welfare, and virtual value functions}

By re-working Myerson's optimal auction theory under the new objective that we want to optimize and the constraints that come with the definition of the allocation rule from an AMM, we get that:

\begin{theorem}
\label{thm:exp_virtual_welfare}
[Expected Profit = Expected Virtual Welfare]
The expected profit according to \Cref{eq:revenue} of a market maker under the assumptions of \Cref{subsec:model_optimal_alloc} is equal to the expected virtual welfare, i.e.,
\[
\E_{\ph\sim\calD}[\text{Profit}] = \E_{\ph\sim\calD} \left[ |x(\ph)| \cdot \left( \phi_u(\ph) \cdot \bone_{\ph\ge\pz} + \phi_l(\ph) \cdot \bone_{\ph\le\pz} \right) \right]
\,
\]
where $\bone_{A}$ is the indicator function that is 1 when $A$ holds and 0 otherwise, and $\phi_u(s), \phi_l(s)$ are the virtual value functions defined by
\begin{align}
\label{eq:upper_virtual_value}
\phi_u(s) \triangleq s - \frac{1-F(s)}{f(s)} -\pt(\pz,s)
\end{align}
and
\begin{align}
\label{eq:lower_virtual_value}
\phi_l(s) \triangleq \pt(\pz,s)-s -\frac{F(s)}{f(s)}
\,.
\end{align}
\end{theorem}

\begin{proof}
By \Cref{eq:ic_amm_payment} and re-writing $x(\ph)$ using the normalization condition, we have that\footnote{This proof assumes differentiability of $x(s)$ for simplicity of the technical exposition but carries through verbatim in terms of Riemann–Stieltjes integrals where $x'(s)~ds = dx(s)$, which always exist because $x(s)$ is guaranteed to be monotonic by \Cref{prop:ic_amms_non_incr_g}.}
\[
y(\ph) = \int_\pz^\ph sx'(s)~ds
\text{ and }
x(\ph) = \int_\pz^\ph x'(s)~ds
\,,
\]
so that the expected profit $\E_{\ph\sim\calD}[\text{Profit}]$ is
\begin{align*}
&\int_\pmin^\pmax \int_\pz^\ph \left( s-\pt(\pz,\ph) \right) x'(s) f(\ph) ~ds~d\ph
\\ =&
\int_\pz^\pmax \int_\pz^\ph \left( s-\pt(\pz,\ph) \right) x'(s) f(\ph) ~ds~d\ph
-
\int_\pmin^\pz \int_\ph^\pz \left( s-\pt(\pz,\ph) \right) x'(s) f(\ph) ~ds~d\ph
\\ =&
\int_\pz^\pmax \int_s^\pmax \left( s-\pt(\pz,\ph) \right) x'(s) f(\ph) ~d\ph~ds
-
\int_\pmin^\pz \int_\pmin^s \left( s-\pt(\pz,\ph) \right) x'(s) f(\ph) ~d\ph~ds
\\ =&
\int_\pz^\pmax x'(s) \left( s(1-F(s)) - \int_s^\pmax \pt(\pz,\ph) f(\ph)~d\ph \right) ~ds
\\ &-
\int_\pmin^\pz x'(s) \left( sF(s) - \int_\pmin^s\pt(\pz,\ph) f(\ph)~d\ph \right) ~ds
\\ =&
-
\int_\pz^\pmax x(s) \left( 1-F(s) + (\pt(\pz,s)-s) f(s) \right) ~ds
+
\int_\pmin^\pz x(s) \left( F(s) + (s-\pt(\pz,s)) f(s) \right) ~ds
\\ =&
\int_\pz^\pmax |x(s)| \left( s-\pt(\pz,s) - \frac{1-F(s)}{f(s)} \right) f(s)~ds
+
\int_\pmin^\pz |x(s)| \left( \pt(\pz,s)-s -\frac{F(s)}{f(s)} \right) f(s)~ds
\,,
\end{align*}
where the third line arises from switching the order of the integrals (and the correct switching depends on the relative order of $\pz$ and $\ph$ which is why the second line is needed; the first integral when $\ph\ge\pz$ and the second vice-versa), the fifth line from integration by parts, and the sixth line because $x(s) \le 0$ for $s \le \pz$ and $x(s) \ge 0$ for $s \ge \pz$.

Therefore, the two virtual value functions are as in the theorem description.
\end{proof}

\subsection{Derivation of the optimal allocation rule and the no-trade gap}
\label{subsec:deriv_general_opt}

We now turn our attention to optimizing the expected profit, i.e., finding the allocation rule (equivalently, the demand curve up to an arbitrary positive additive constant) which maximizes the expected profit.

\begin{theorem}
\label{thm:opt_alloc_rule}
The optimal allocation rule $x^\star(\ph)$ that maximizes the expected profit according to \Cref{eq:revenue} under the assumptions of \Cref{subsec:model_optimal_alloc} has the form
\begin{align*}
x^\star(\ph) = \begin{cases}
1,  &\text{for } p_1 \le \ph \le \pmax \\
0,  &\text{for } p_2 < \ph < p_1 \\
-1, &\text{for } \pmin \le \ph \le p_2
\end{cases}
\,,
\end{align*}
where $p_1 \ge \pz$ and $p_2 \le \pz$ are some roots of the upper and lower virtual value functions $\phi_u(s), \phi_l(s)$ of \Cref{eq:upper_virtual_value,eq:lower_virtual_value} respectively.
\end{theorem}
\begin{proof}
We will use the equivalence of expected profit and expected virtual welfare established by \Cref{thm:exp_virtual_welfare}.
Because we operate under the constraints that $x(\pz)=0$, $x(\ph) \ge 0$ for $\ph \ge \pz$, and $x(\ph) \le 0$ for $\ph \le \pz$, we can just equivalently optimize the two integrals arising from the right hand side of \Cref{thm:exp_virtual_welfare} separately to obtain the optimal allocation rule.
Additionally, since $\phi_u(\pz) = -\frac{1-F(\pz)}{f(\pz)} \le 0$, $\phi_u(\pmax) = \pmax-\pt(\pz,\pmax) \ge 0$, $\phi_l(\pmin) = \pt(\pz,\pmin)-\pmin \ge 0$, and $\phi_l(\pz) = -\frac{F(\pz)}{f(\pz)} \le 0$, there always exists a non-trivial optimal solution to the optimization problem.

If $\phi_u(s)$ was monotone (increasing), then by the above sign changes, there would exist a (unique) root $p_1$ such that $\phi_u(p_1)=0$.
A simple argument by contradiction then shows that the structure of the optimal allocation rule is to allocate the maximum allowable quantity of asset (by the constraint of unit demand, $x(s)\le 1$) to be sold from the price $p_1$ onwards; that is,
\begin{align}
\label{eq:upper_optimal_x}
x^\star(s) = \begin{cases}
1, &\text{for } p_1 \le s \le \pmax \\
0, &\text{for } \pz \le s < p_1
\end{cases}
\,.
\end{align}

For the lower virtual value function, a similar argument shows that, if $\phi_l(s)$ was monotone (decreasing), and $p_2$ is its unique root (which is guaranteed to exist by the aforementioned sign changes), then the optimal allocation rule obeying the constraint of unit demand $x(s)\ge -1$ would be
\begin{align}
\label{eq:lower_optimal_x}
x^\star(s) = \begin{cases}
0,  &\text{for } p_2 < s \le \pz \\
-1, &\text{for } \pmin \le s \le p_2
\end{cases}
\,.
\end{align}

In full generality, if the virtual value functions are not monotone, then there may now be multiple roots of those functions. In this case, and with an argument by contradiction that if the allocation rule were to switch in between roots then we could obtain better or equal expected profit by modifying the allocation rule to remain constant in between roots, along with the argument that if the allocation rule were not to allocate the entire available demand (1, $-1$, respectively) then we could obtain better or equal expected profit by increasing (or decreasing, respectively) the allocated quantity of the asset until the limit (of 1 or $-1$, respectively), we obtain that the optimal allocation rule is exactly characterized by \Cref{eq:lower_optimal_x,eq:upper_optimal_x}, where this time the choice of the roots as $p_1$ (and $p_2$, respectively) needs to be the root of the upper (lower, respectively) virtual value function that gives the maximum expected profit under the allocation rule that has the form of \Cref{eq:upper_optimal_x} (\Cref{eq:lower_optimal_x}, respectively). Therefore, the optimal allocation rule is as per the theorem statement.
\end{proof}

Therefore, in any case, we observe that the optimal solution of \Cref{eq:upper_optimal_x,eq:lower_optimal_x} has a well-specified \textbf{no-trade gap} around $\pz$ where the market maker is not willing to buy or sell any quantity of the asset to the trader.
More specifically, the length of this no-trade gap is precisely $p_1-p_2$, where $p_1\ge\pz$, $p_2\le\pz$ are as per the above analysis roots of the upper and lower virtual value functions respectively.

We can now specialize in the two interesting extreme cases that were examined in \Cref{subsec:intro_specific} to obtain the respective observations there: first, consider the pure noise trading case, where the traders are assumed to offer no information on the true asset price, and hence the market maker's prior is static, i.e., in the case that the update rule is $\pt(\pz,\ph)=\pz,\ \forall\ph\in[\pmin,\pmax]$. In this case, we see that there is a \textit{non-trivial} no-trade gap, since we have for the extremal values that $\phi_u(\pz) < 0$, $\phi_u(\pmax) > 0$, $\phi_l(\pmin) > 0$, and $\phi_l(\pz) < 0$ (see the beginning of \Cref{subsec:deriv_general_opt} for the expressions), hence the respective roots are contained strictly in between their corresponding intervals (by continuity of the virtual value functions).
This no-trade gap is akin to a bid-ask spread observed in traditional CLOB-based markets, and is attributable to the monopoly pricing power of the market maker.
In addition, it is \textit{distribution-dependent} and its length can be intuitively understood as being smaller under distributions $\calD$, according to which the traders' estimates $\ph$ are distributed, that are more ``concentrated'' around $\pz$\footnote{With this statement, we intend to only give intuition around the length of the interval, because there can be pathological cases of pairs of distributions that do not follow this generic intuition, and hence it need not be true in full generality.}. To mention an example that illustrates this intuition, the no-trade gap under pure noise trading and a uniform distribution supported on $[p_\text{min}^{(1)}, p_\text{max}^{(1)}]$ will be smaller than the one under another uniform distribution supported on $[p_\text{min}^{(2)}, p_\text{max}^{(2)}]$, where $p_\text{min}^{(1)} > p_\text{min}^{(2)} \text{ and } p_\text{max}^{(1)} < p_\text{max}^{(2)}$.

At the other extreme, consider the case where the traders always have perfect information about the true asset price, i.e., the market maker's update rule is $\pt(\pz,\ph) = \ph,\ \forall\ph\in[\pmin,\pmax]$. In this case, we observe that the two virtual value functions become
\[
\phi_u(s) = -\frac{1-F(s)}{f(s)} \le 0
\text{ and }
\phi_l(s) = -\frac{F(s)}{f(s)} \le 0
\,,
\]
thus the optimal allocation rule is everywhere zero: $x^\star(\ph) = 0,\ \forall \ph\in[\pmin,\pmax]$.
It is therefore the case that this corollary which may be phrased as a version of a \textit{no-trade theorem} in the presence of perfect adverse selection is \textit{distribution-independent}: that is, regardless of the distribution $\calD$ of the traders' asset price estimates $\ph\sim\calD$, if the market maker believes that the traders have perfect information, then their optimal allocation rule is to perform no-trade at all. More specifically:

\begin{corollary}
[No-trade theorem when all traders have perfect information]
If the update rule $\pt(\pz,\ph) = \ph,\ \forall\ph\in[\pmin,\pmax]$, then for every distribution $\calD$, the optimal allocation rule is $x^\star(\ph) = 0,\ \forall\ph\in[\pmin,\pmax]$.
\end{corollary}

\section{An important special case: linear instantiation of the update rule}
\label{sec:lin_lambda}

We move on to consider an interesting special case of the update rule of \Cref{subsec:model_optimal_alloc}: the simplest specialization that allows us to capture an arbitrary mixture of both noise trading and adverse selection, and enables us to show exactly how they intermingle with one another.
More specifically, we consider the case when the update rule is a convex combination of the two above extreme cases, controlled by a parameter $\lambda\in[0,1]$:
\begin{align}
\label{eq:lin_upd_rule}
\pt(\pz,\ph) = \lambda \pz + (1-\lambda) \ph
\end{align}

The extreme cases arise as follows: for $\lambda=1$, we obtain that the update remains the same as the market maker's prior ($\pz$), i.e., the pure noise trading model, where the trader is assumed to provide no information on the asset price.
For $\lambda=0$, we obtain that the update always follows exactly the trader's estimate ($\ph$), thus the trader is considered as having perfect information around the asset price, i.e., pure adverse selection is attained.

There are standard normality assumptions under which this linear update rule naturally emerges as the update rule for Bayesian updating. For example, if we consider that the asset value $V\sim\calN(p_0, \sigma_0^2)$ where $\calN$ is the normal distribution, and the trader's estimate is $\ph = V + \eps$, where $\eps\sim\calN(0,\sigma_\eps^2)$, i.e., $\ph\sim\calN(p_0,\sigma_0^2+\sigma_\eps^2)$, then the posterior according to a Bayesian update step would also be a normal distribution with mean $\pt(\pz,\ph) = \lambda \pz + (1-\lambda) \ph$, where $\lambda=\frac{\sigma_\eps^2}{\sigma_0^2+\sigma_\eps^2}$ \parencite{bayesian}.
Of course, this is just an illustration of the motivation behind the linear update rule; it ignores the fact that asset prices cannot obtain negative values, but a similar calculation can be performed for log-normal distributions, for instance, and also holds approximately for truncated Gaussian distributions in the case of the standard deviation being small compared to the mean value.
However, in this section, we will not be necessarily making any such Gaussian assumptions; we will only be assuming that the conducted update corresponds to a linear interpolation (in fact, convex combination) of the two cases of $\pz$ and $\ph$ that we have seen in \Cref{subsec:intro_specific} as indicative of noise trading and adverse selection.

In order to proceed with a more precise characterization of the interval of no-trade around $\pz$, we will make a similar assumption to the regularity of the distribution $\calD$ made by \textcite{myerson_optimal_auctions}:
\begin{assumption}
\label{ass:reg}
The upper and lower virtual value functions, obtained by plugging in \Cref{eq:lin_upd_rule} to \Cref{eq:upper_virtual_value,eq:lower_virtual_value},
\[
\phi_u(s;\lambda) = \lambda(s-\pz) - \frac{1-F(s)}{f(s)}
\text{ and }
\phi_l(s;\lambda) = \lambda(\pz-s) - \frac{F(s)}{f(s)}
\]
are non-decreasing and non-increasing (respectively) with respect to $s$ for all $\lambda\in[0,1]$.
\end{assumption}

\paragraph{Discussion of \Cref{ass:reg}.}
\Cref{ass:reg} can be intuitively understood as stating that the distribution of $\ph$ (i.e., $\calD$) does not have heavy tails away from $\pz$ on either direction (i.e., left or right tails). We now explain in a detailed way why this is the case and why this assumption is standard and naturally motivated.
First, observe that the original regular distributional assumption of \textcite{myerson_optimal_auctions} is equivalent to the assumption that $\phi_u(s;\lambda)$ is strictly increasing for $\lambda=1$.
Additionally, notice that an equivalent re-statement of \Cref{ass:reg} is that the respective monotonicities hold for $\lambda=0$, since this implies the monotonicities for all $\lambda\in[0,1]$.
In particular, for the monotonicity of the upper virtual value function, we observe that it exactly corresponds to having a monotone hazard rate, where the hazard rate of a distribution is defined as
\[
h(s) \triangleq \frac{f(s)}{1-F(s)}
\,.
\]
Distributions that have a monotone hazard rate (MHR) have been widely used before as common distributional assumptions in revenue maximizing settings \parencite{hartlinebook,coleshravas2017,cole2014sample,zora63129}. More generally, $\phi_u(s;\lambda)$ being non-decreasing for a specific $\lambda$ corresponds to the notion of an $\alpha$-strongly regular distribution (Definition 1.1) of \textcite{cole2014sample} where $\alpha=1-\lambda$.
As for the monotonicity of the lower virtual value function assumption, we observe that if we make the change of variables $p' = \pz - \ph$ (where $\ph\sim\calD$ is our original random variable, and $p'$ is the new one), then $\phi_l(s;0)$ being non-increasing would be equivalent to the distribution of $p'$ having a monotone hazard rate (MHR), because if we denote by $G(s)$ and $g(s)$ the CDF and the PDF of $p'$ respectively, then from standard probabilities, we get that $G(s) = 1-F(\pz-s)$, and $g(s) = f(\pz-s)$, hence
\[
\frac{g(s)}{1-G(s)} = \frac{f(\pz-s)}{F(\pz-s)}
\,,
\]
and $\phi_l(s;0)=-\frac{F(s)}{f(s)}$ being non-increasing is equivalent to $\frac{f(\pz-s)}{F(\pz-s)}$ being non-decreasing.

Intuitively, distributions satisfying \Cref{ass:reg} are well-behaved, in the sense that they correspond to distributions with MHR on both of their tails, i.e., distributions that have both left and right tails that are no heavier than an exponential.

Many interesting distributions satisfy the ``non-heavy-tailed'' property that \Cref{ass:reg} refers to. For example, all uniform and exponential distributions satisfy the assumption. More generally, all distributions with log-concave densities have a monotone hazard rate \parencite{zora63129}, and if centered around $\pz$ (for instance symmetrically on both sides), both tail ends would not be heavy and they would satisfy the assumption.
As with classical non-MHR distributions, \Cref{ass:reg} is not satisfied by most heavy-tailed distributions; the difference is that this assumption includes both left and right tails of the distribution.
A good example of a double-sided (left and right) heavy-tailed distribution that would fail to satisfy \Cref{ass:reg} would be the Cauchy distribution with a location parameter of $\pz$, for which the moment generating function does not exist anywhere: in this case, both of its left and right tails around $\pz$ are heavy \parencite{cauchy}.

We now move on to characterize the length of the no-trade interval present in the optimal solution of \Cref{eq:upper_optimal_x,eq:lower_optimal_x} as $\lambda$ varies.

\begin{proposition}
\label{prop:lin_length_no_trade}
The length of the no-trade gap in the optimal solution that obtains the maximum expected profit under \Cref{eq:lin_upd_rule} is a decreasing function of $\lambda$, when the distribution $\calD$ satisfies \Cref{ass:reg}.
\end{proposition}
\begin{proof}
First, observe that under \Cref{ass:reg}, the optimal allocation rule is exactly characterized by \Cref{eq:upper_optimal_x,eq:lower_optimal_x}, whereby there exists exactly one root for each of the corresponding upper and lower virtual value functions.
We now need to show how the roots move, while $\lambda$ varies. In particular, we will prove that the root of the upper virtual value function increases as $\lambda$ decreases; the proof that the root of the lower virtual value function decreases is similar. These two combined facts mean that, as $\lambda$ decreases, the length of the interval of no-trade increases (because each of the roots moves away from $\pz$), thus proving the desired result.

We finish the proof as follows: consider $\lambda_1 \le \lambda_2$. The roots $p_1, p_2$ of the upper virtual value functions $\phi_u(s;\lambda_1), \phi_u(s;\lambda_2)$ for $\lambda_1, \lambda_2$ respectively satisfy
\[
\phi_u(p_1;\lambda_1) = 0
\text{ and }
\phi_u(p_2;\lambda_2) = (\lambda_2-\lambda_1)(p_2-\pz) + \phi_u(p_2;\lambda_1) = 0
\,.
\]

Additionally, $(\lambda_2-\lambda_1)(p_2-\pz) \ge 0$ since $\lambda_1 \le \lambda_2$ and $p_2 \ge \pz$ (by definition of the upper virtual value function), so the latter equality means that $\phi_u(p_2;\lambda_1) \le 0 = \phi_u(p_1;\lambda_1)$.
Therefore, by the monotonicity of $\phi_u(s;\lambda_1)$, the latter is equivalent to the ordering $p_2 \le p_1$ of the respective roots, as desired.
\end{proof}

\Cref{prop:lin_length_no_trade} allows us to \emph{make precise the important interpretation of how the effects of monopoly pricing and adverse selection interplay} in a model of trading that incorporates both noise trading and better-informed traders.
More specifically, in the optimal solution to the profit maximization problem of the market maker, there is a minimal interval of no-trade (corresponding to the pure noise trading case, $\lambda=1$) that exhibits the monopoly pricing power being the sole driving force of the shape of the optimal mechanism.
As $\lambda$ gets smaller, and the market maker's model for the trader pre-supposes more and more traders with perfect information, according to \Cref{prop:lin_length_no_trade}, the interval around $\pz$ on which the AMM is not willing to trade grows larger and larger.
This change is attributed to adverse selection being added in the mixture because of the traders that are better informed than the market maker.
At the other extreme, where pure adverse selection is observed ($\lambda=0$), the optimal solution is to provide no liquidity, or equivalently, perform no trade. This is a distribution-independent result that is due to the corresponding virtual value functions (both upper and lower) being purely non-positive in that case (c.f., the expressions of \Cref{ass:reg} for $\phi_u(s;0)$ and $\phi_l(s;0)$ where we substitute $\lambda=0$), making the optimal solution that obtains the maximum profit be exactly the no-trade-anywhere solution $x(\ph) = 0,\ \forall \pmin\le\ph\le\pmax$.

\section{Conclusion and future directions}

The goal of this work is to isolate the most essential and salient properties of the problem of optimal liquidity provision: the zero-sum game between traders and liquidity providers, uncertain valuations, asymmetric information, and price updates conditioned on trade.
We show how a ``no-trade gap'' (that arises as a ``bid-ask spread'' in traditional market microstructure literature for reasons other than the ones considered here)  surprisingly arises in very general settings of liquidity provision in exchange mechanisms, like the ones we consider, allowing us to unlock a fruitful connection between optimal market-making and the Myersonian theory of optimal auction design. This connection has not been made in prior literature to the best of our knowledge.

Our general framework for reasoning about AMMs is based on a Bayesian-like belief inference procedure. Such Bayesian updating procedures are well-known and widely used in machine learning, even in the specific context of market-making \parencite{nips2008sequential_marketmaking}.
The generality of our framework has the exact intention of capturing any possible exchange mechanism, including the most prominent examples of both limit order books and automated market makers \parencite[see also][]{exchange_complexity}.
At the same time, for the optimal allocation rule whose structure we hereby characterize, it would be interesting to give the complexity of the market maker's computation with update rules other than the linear one, which we gave as an indicative example of the mixing effects of both adverse selection and noise trading.
In the case of update rules that can arise from an online learning framework with iterative belief updates, what makes the complexity determination not straightforwardly arise from these online learning frameworks is exactly the needed calculation of the roots in \Cref{thm:opt_alloc_rule}.

It might seem, at a first glance, that the market maker in the framework given in this work is not just risk-neutral but also myopic, i.e., maximizes expected profit from the next trade given its update rule, belief distribution over traders’ estimates, and its own current asset value estimate.
However, our framework can capture much more generic structures of dependencies that allow the market maker to not be myopic, including a (possibly discounted) cumulative profit/reward structure.
It would be an interesting future direction to examine whether the insights we give in this work change in such cases; we conjecture that they do not, but additional effects could be observed.
In fact, our Bayesian (update rule) framework can be adjusted to accommodate these more complicated structures, where the updated belief (through the update rule) can have any future expectations already baked in.
In particular, sequentially, the market maker’s new updated belief of the prior round of trading becomes their prior belief in the next trade, so that the prior at each round depends on the entire past history of trades; the update rule at each time holds the best (future) expectation of the market maker.

Finally, future empirical work that probes the model's assumptions and results could be explored, just as empirical work on revenue-maximizing auctions and online learning has served to usefully probe Myerson's original theoretical framework.

\section*{Acknowledgments} 
The first author is supported in part by NSF awards CNS-2212745, CCF-2212233, DMS-2134059, and CCF-1763970.
The second author is supported by the Briger Family Digital Finance Lab at Columbia Business School.
The third author's research at Columbia University is supported in part by NSF awards CNS-2212745, and CCF-2006737.

\section*{Disclosures}
The second author is an advisor to fintech companies.
The third author is Head of Research at a16z crypto, which reviewed a draft of this article for compliance prior to publication and is an investor in various decentralized finance projects, including Uniswap, as well as in the crypto ecosystem more broadly (for general a16z disclosures, see \url{https://www.a16z.com/disclosures/}).

\printbibliography

\appendix

\end{document}